\documentclass{new_tlp}
\usepackage{amsmath}
\usepackage{mathptmx}
\usepackage[ruled,linesnumbered,vlined,nok wfunc]{algorithm2e}
\usepackage{xspace}
\usepackage{color}
\usepackage{booktabs}
\usepackage{wasysym}
\usepackage{comment}
\usepackage{thm-restate}
\usepackage{todonotes}
\usepackage{url}

\newcommand\bcmdtab{\noindent\bgroup\tabcolsep=0pt%
  \begin{tabular}{@{}p{10pc}@{}p{20pc}@{}}}
\newcommand\ecmdtab{\end{tabular}\egroup}

\newcommand{\wasp}{\textsc{wasp}\xspace}
\newcommand{\waspict}{\textsc{wasp-ict}\xspace}
\newcommand{\waspor}{\textsc{wasp-or}\xspace}
\newcommand{\waspone}{\textsc{wasp-one}\xspace}
\newcommand{\waspopt}{\textsc{wasp-opt}\xspace}
\newcommand{\waspcm}{\textsc{wasp-cm}\xspace}
\newcommand{\clasp}{\textsc{clasp}\xspace}

\def\A{\ensuremath{\mathcal{A}}\xspace}

\newtheorem{example}{Example}[section]

\usepackage{pgfplots}\usetikzlibrary{plotmarks}

\pgfplotsset{
	filter discard warning=false 
	, legend cell align=left
	, minor grid style={loosely dotted, lightgray}
	, major grid style={loosely dashed, lightgray}
}

\def\naf{\ensuremath{\raise.17ex\hbox{\ensuremath{\scriptstyle\mathtt{\sim}}}}\xspace}
 \title[Cautious reasoning in ASP via minimal models and unsatisfiable cores]
       {Cautious reasoning in ASP \\via minimal models and unsatisfiable cores}

  \author[M. Alviano et al.]
{
	MARIO ALVIANO\\
	DEMACS, University of Calabria, Italy\\
	\email{alviano@mat.unical.it}\\
	\and CARMINE DODARO\\
	DIBRIS, University of Genova, Italy\\
	\email{dodaro@dibris.unige.it}\\
	\and MATTI J\"ARVISALO\\
	HIIT, Department of Computer Science, University of Helsinki, Finland\\
	\email{matti.jarvisalo@helsinki.fi}\\
	\and MARCO MARATEA\\	
	DIBRIS, University of Genova, Italy\\
	\email{marco@dibris.unige.it}\\
	\and ALESSANDRO PREVITI\\
	HIIT, Department of Computer Science, University of Helsinki, Finland\\
	\email{alessandro.previti@helsinki.fi}
}

\jdate{March 2003}
\pubyear{2003}
\pagerange{\pageref{firstpage}--\pageref{lastpage}}
\doi{S1471068401001193}

\begin{document}

\label{firstpage}

\maketitle

  \begin{abstract}

Answer Set Programming (ASP) is a logic-based knowledge representation framework, 
supporting---among other reasoning modes---the central task of query answering.
In the propositional case, query answering amounts to computing cautious consequences of the input program among the atoms in a given set of candidates, where a cautious consequence is an atom belonging to all stable models.
Currently, the most efficient algorithms either iteratively verify the existence of a stable model of the input program extended with the complement of one candidate, where the candidate is heuristically selected, or introduce a clause enforcing the falsity of at least one candidate, so that the solver is free to choose which candidate to falsify at any time during the computation of a stable model.
This paper introduces new algorithms for the computation of cautious consequences, with the aim of driving the solver to search for stable models discarding more candidates.
Specifically, one of such algorithms enforces minimality on the set of true candidates, where different notions of minimality can be used, and another takes advantage of unsatisfiable cores computation.
The algorithms are implemented in \textsc{wasp}, and experiments on benchmarks from the latest ASP competitions show that the new algorithms perform better than the state of the art.
(Under consideration for acceptance in TPLP).
  \end{abstract}


\section{Introduction}

Answer set programming (ASP; \citeNP{DBLP:journals/amai/Niemela99,mare-trus-99,DBLP:books/daglib/0040913,DBLP:journals/cacm/BrewkaET11,DBLP:journals/aim/JanhunenN16,DBLP:journals/aim/Lifschitz16}) 
is today a popular logic-based knowledge representation and constraint programming framework with constructs specifically designed for industrial applications \cite{DBLP:conf/aaai/GebserKROSW13,DBLP:journals/tplp/DodaroGLMRS16}.
In ASP, knowledge is expressed via logic programs, also referred to as  ASP programs, or simply programs.
A program encodes a particular problem domain of interest. 
The semantic is given by the so-called stable models 
(or answer sets; \citeNP{gelf-lifs-88,DBLP:journals/ngc/GelfondL91}), which represent the plausible scenarios for the input instance.
As a logic-based knowledge representation framework, ASP systems support several computational tasks, 
many of them characterized by higher complexities than the corresponding problems in classical propositional logic~\cite{DBLP:journals/tods/EiterGM97}.
One  central task is  query answering, which in ASP is supported also in presence of uninterpreted function symbols \cite{DBLP:journals/tplp/AlvianoFL10}.
Query answering can be addressed with two reasoning modes, namely brave and cautious reasoning, providing answers witnessed by some or all stable models, 
respectively.
In particular, cautious reasoning over ASP programs has a significant number of interesting applications, including consistent query answering
\cite{DBLP:journals/tplp/ArenasBC03,DBLP:journals/tplp/MannaRT13}, data integration \cite{DBLP:conf/lpnmr/Eiter05}, multi-context systems \cite{DBLP:conf/ijcai/BrewkaRS07}, and ontology-based reasoning \cite{DBLP:journals/ai/EiterILST08}.

Focusing on the propositional level on which the stable model search algorithm of ASP systems typically operates on, query answering amounts to computing \emph{cautious consequences} of the input program among the atoms in a given set of candidates, where cautious consequences are atoms contained in all stable models.
At a symbolic level, more complex queries can be expressed by means of rules defining intentional predicates, whose ground instances form the set of candidates.
For example, a union of conjunctive queries requires a rule for each conjunction in the union, the body being the conjunction, and the head using a predicate not occurring in the program and collecting all selected variables.

Currently, the most efficient algorithms for cautious reasoning in ASP, as the ones implemented by \clasp~\cite{DBLP:journals/ai/GebserKS12}
 and \wasp~\cite{DBLP:conf/cilc/DodaroAFLRS11,DBLP:conf/lpnmr/AlvianoDFLR13,DBLP:conf/lpnmr/AlvianoDLR15}, are mainly adaptations of techniques proposed for the computation of backbones in the context of propositional logic~\cite{DBLP:journals/aicom/JanotaLM15}.
The common idea behind such algorithms is to perform several iterative calls to an ASP oracle for computing a stable model of the input program subject to algorithm-dependent additional constraints. 
The added constraints inferred from each of the computed stable models are guaranteed to discard some atoms from the set of candidates.
Specifically, state-of-the-art algorithms for cautious reasoning in ASP either (i)~iteratively verify the existence of a stable model of the input program extended with the complement of 
one candidate, where the candidate is heuristically selected, or 
(ii)~introduce an integrity constraint that enforces 
the falsity of at least one candidate, so that the solver is free to choose which candidate to falsify at any time during the computation of a stable model.
Efficient computation is achieved thanks to the incremental stable model search procedure implemented by modern solvers, which possibly reuses several learned clauses from previous computations in order to avoid redundantly repeating some of the already performed computations.

This paper introduces new algorithms for determining the cautious consequences of ASP programs.
The intuition behind the proposed algorithms is to drive the search for stable models so that a high number of candidate atoms is discarded at a time.
In more detail, two different algorithmic approaches are proposed, based on \emph{minimal models} and \emph{unsatisfiable cores}.
The first of the algorithmic approaches introduced in the paper enforces 
minimality on the set of true candidates. This approach in fact constitutes  a meta-algorithm that can be 
instantiated with several minimality notions and different minimality procedures;
two instantiations are presented in this paper, one based on the computation of
subset-minimal models via qualitative preferences
\cite{DBLP:journals/constraints/RosaGM10}, and the other based on the computation of cardinality-minimal models via algorithm \textsc{one} \cite{DBLP:conf/ijcai/AlvianoDR15}
recently proposed in the context of maximum satisfiability (MaxSAT) solving.

The second algorithm for cautious reasoning introduced in this paper takes advantage of unsatisfiable cores computation, which has already proven 
very effective for solving ASP optimization problems \cite{DBLP:journals/tplp/AlvianoD16}.
For some intuition on this approach,
modern ASP solvers take as input a program and a set of assumption literals, and search for a stable model of the given program extending the assumption literals.
If such a stable model does not exist, a set of assumption literals that is sufficient to cause the failure of the search is returned.
Such a sufficient set of assumption literals constitutes an
 unsatisfiable core in the context of ASP optimization \cite{DBLP:conf/ijcai/AlvianoD17}.
Accordingly, an atom is a cautious consequence of a given program if and only if the negation of the atom is an unsatisfiable core.
Hence, the new algorithm searches for stable models falsifying all candidates, with the aim of eliminating all remaining candidates at once. As soon as
no stable model under the current set of assumptions exists, the returned unsatisfiable core is either minimized to a singleton or used to discard candidates.

The new algorithms are implemented in \textsc{wasp}, which provides  a fair basis for comparing the algorithms proposed 
in this paper with the two state-of-the-art cautious reasoning approaches previously implemented in \textsc{wasp}~\cite{DBLP:journals/tplp/AlvianoDR14}.
Indeed, in this way  all ASP oracle calls are realized  by the incremental stable model search procedure provided by \textsc{wasp}, 
enabling a clean comparison of  the  behavior and effectiveness of the algorithms.
The implemented algorithms are empirically tested on benchmarks from the latest ASP competitions \cite{DBLP:journals/ai/CalimeriGMR16,DBLP:journals/jair/GebserMR17}.
The experiments show that the  algorithms proposed in this paper perform better than the state of the art on the available test cases.
Moreover, the efficiency of the new algorithms is confirmed via a comparison with the state-of-art ASP solver \textsc{clasp} \cite{DBLP:journals/ai/GebserKS12}, which also supports
cautious reasoning. While \textsc{clasp}  is known to be typically faster 
than \textsc{wasp} in deciding the existence of stable models (i.e., satisfiability checking), the new cautious reasoning algorithms proposed in this paper,
as implemented in \textsc{wasp}, outperform \textsc{clasp} on the task of cautious reasoning.

This paper is structured as follows. 
Section~\ref{sec:prel} provides necessary background on ASP, including syntax and semantics, and
Section~\ref{sec:sota} includes an overview of existing state-of-the-art algorithms for cautious reasoning. The new algorithms for cautious reasoning are detailed in Section~\ref{sec:newalgo},
and Section~\ref{sec:exp} provides results from an empirical comparison of the new algorithms
and earlier proposed approaches. Before conclusions, 
related work is discussed in Section~\ref{sec:related}.

\section{Preliminaries}
\label{sec:prel}
This section overviews preliminaries on ASP and
the main steps of stable model search as necessary for the rest of the paper.

\paragraph{Syntax.}
Let \A be a fixed, countable set of (propositional) \emph{atoms} including $\bot$.
A \emph{literal} $\ell$ is either an atom $p$, or its negation $\naf p$, where \naf denotes \emph{default negation}.
Let $\overline{\ell}$ denote the complement of $\ell$, i.e., $\overline{p} := \naf p$, and $\overline{\naf p} := p$, for all $p \in \A$.
For a set $L$ of literals, $\overline{L} := \{\overline{\ell} \mid \ell \in L\}$, $L^+ := L \cap \A$, and $L^- := \overline{L} \cap \A$.
A \emph{program} $\Pi$ consists of a set of \emph{disjunctive} and \emph{choice rules}, which have respectively the following forms:
\begin{align}
    \label{eq:rule}
        p_1 \vee \cdots \vee p_m \leftarrow \ell_1, \ldots, \ell_n \\
    \label{eq:choice}
    \{\ell_1, \ldots, \ell_n\} \geq k
\end{align}
where $p_1, \ldots, p_m$ are distinct atoms, $\ell_1, \ldots, \ell_n$ are distinct literals, and $k,m,n \geq 0$.
For a rule $r$ of the form (\ref{eq:rule}), the disjunction $p_1 \vee \cdots \vee p_m$ is called the \emph{head} of $r$, denoted $H(r)$;
the conjunction $\ell_1, \ldots, \ell_n$ is the \emph{body} of $r$, denoted $B(r)$.
With a slight abuse of notation, $H(r)$ and $B(r)$ also denote the sets of their elements.
A \emph{constraint} is a rule of the form (\ref{eq:rule}) such that $H(r) \subseteq \{\bot\}$. 
For a rule $r$ of the form (\ref{eq:choice}), let $\mathit{lits}(r)$ be $\{\ell_1,\ldots,\ell_n\}$, and $\mathit{bound}(r)$ be $k$. 

\paragraph{Semantics.}
An \emph{interpretation} $I$ is a subset of $\A \setminus \{\bot\}$;
atoms in $I$ are assigned true, and those in $\A \setminus I$ are assigned false.
$I$ is a model of a rule $r$ of the form (\ref{eq:rule}), denoted $I \models r$, if $H(r) \cap I \neq \emptyset$ whenever $B(r)^+ \subseteq I$ and $B(r)^- \cap I = \emptyset$.
$I$ is a model of a rule $r$ of the form (\ref{eq:choice}), denoted $I \models r$, if $|\mathit{lits}(r)^+ \cap I| + |\mathit{lits}(r)^- \setminus I| \geq \mathit{bound}(r)$ holds.
$I$ is a model of a program $\Pi$, denoted $I \models \Pi$, if it is a model of all rules in $\Pi$.
The definition of stable model is based on the following notion of reduct.
Let $\Pi$ be a program, and $I$ be an interpretation.
The \emph{reduct} of $\Pi$ with respect to $I$, denoted $\Pi^I$, is obtained from $\Pi$ by deleting each rule $r$ of the form (\ref{eq:rule}) such that $B(r)^- \cap I \neq \emptyset$, by replacing each rule $r$ of the form (\ref{eq:choice}) with rules of the form $p \leftarrow$ for all $p \in \mathit{lits}(r) \cap I$, and removing negated atoms in the remaining rules.
An interpretation $I$ is a \emph{stable model} of $\Pi$ if $I \models \Pi^I$ and there is no $J \subset I$ such that $J \models \Pi^I$.
Let $\mathit{SM}(\Pi)$ denote the set of stable models of $\Pi$.
If $\mathit{SM}(\Pi) \neq \emptyset$ then $\Pi$ is \emph{coherent}, otherwise it is \emph{incoherent}. 
An atom $p \in \A$ is a \emph{cautious consequence} of a program $\Pi$ if $p$ belongs to all stable models of $\Pi$.
The set of cautious consequences of $\Pi$ is denoted $\mathit{CC}(\Pi)$.

\begin{example}[Running example]\label{ex:program}
Let $\Pi_\mathit{run}$ be the following program:
\begin{align*}
\begin{array}{llll}
r_1: \quad a \ \vee \	b \leftarrow	\qquad& r_2: \quad c \ \vee \ d  	\leftarrow \qquad& r_3: \quad q_1 \leftarrow a 		 	  \qquad	   & 	r_4: \quad  q_1 \leftarrow b\\
r_5: \quad q_2 \leftarrow c 			\qquad & r_6: \quad q_3 \leftarrow \naf 	c \qquad &	r_7: \quad q_3 \leftarrow \naf d 		\qquad 		& 	r_8: \quad q_4 \leftarrow d.
\end{array}
\end{align*}
$\mathit{SM}(\Pi_\mathit{run})$ comprises the following stable models:
$I_1 := \{a,d,q_1,q_3,q_4\}$, $I_2 := \{a,c,q_1,q_2,q_3\}$, $I_3 := \{b, d, q_1, q_3, q_4\}$, and $I_4 := \{b, c, q_1, q_2, q_3\}$.
Therefore, the set $\mathit{CC}(\Pi_\mathit{run})$ of cautious consequences of $\Pi_\mathit{run}$ is $\{q_1, q_3\}$.
\hfill$\lhd$
\end{example}

\paragraph{Stable model search.}
State-of-the-art algorithms for computing stable models of a given ASP program $\Pi$
extend the
conflict-driven clause learning (CDCL) algorithm \cite{DBLP:journals/tc/Marques-SilvaS99} with ASP-specific
search techniques \cite{DBLP:journals/aim/KaufmannLPS16}.
The input of the algorithm comprises a propositional program $\Pi$, and a set $A$ of literals, called \textit{assumption literals} (or simply assumptions).
Its output is a pair $(I,\emptyset)$, where $I$ is a stable model of $\Pi$ such that $A^+ \subseteq I$ and $A^- \cap I = \emptyset$, if such an $I$ does exist.
Otherwise, if there is no $I \in SM(\Pi)$ such that  $A^+ \subseteq I$ and $A^- \cap I = \emptyset$, the algorithm returns as output a set $(\bot, C)$, where $C \subseteq A$ is such that $\mathit{SM}(\Pi \cup \{\bot \leftarrow \overline{\ell} \mid \ell \in C\}) = \emptyset$;
such a set $C$ is called an \emph{unsatisfiable core} of $\Pi$ with respect to $A$.
In what follows, $\mathrm{ComputeStableModel}(\Pi,A)$ denotes a call to the stable model search algorithm.

\begin{example}[Continuing Example~\ref{ex:program}]\label{ex:assum}
$\mathrm{ComputeStableModel}(\Pi_\mathit{run}, \{b, \naf q_2\})$ necessarily returns $(I_3,\emptyset)$, where $I_3$ is the stable model $\{b, d, q_1, q_3, q_4\}$, while $\mathrm{ComputeStableModel}(\Pi_\mathit{run}, \{\naf q_1, \naf q_2\})$ must return $(\bot,C)$, where $C$ is either the unsatisfiable core $\{\naf q_1\}$, or the unsatisfiable core $\{\naf q_1, \naf q_2\}$.
\hfill$\lhd$
\end{example}

\section{Existing algorithms for cautious reasoning in ASP}
\label{sec:sota}

\begin{algorithm}[t]
	\caption{CautiousReasoning(\textsc{refine}, $\Pi$,$Q$)}\label{alg:cautious}
	
	$U := \emptyset;\quad O := Q$\;
	$(I, C) := \mathrm{ComputeStableModel}(\Pi, \emptyset)$\; 
	\lIf{$I = \bot$}{
		\Return{$\bot$\;}
	}
	$O$ := $O \cap I$\;
	\lWhile{$U \neq O$}{
	    \textsc{refine}\tcp*{\textsc{ict}, \textsc{or}, \textsc{min}, \textsc{cm}, or other procedure}
	}
	\Return{$U$}\;
\end{algorithm}
Algorithm~\ref{alg:cautious} provides a common skeleton shared by several algorithms for computing cautious consequences of ASP programs \cite{DBLP:journals/tplp/AlvianoDR14}.
These algorithms take as input a program $\Pi$ and a set of atoms $Q$, and produce as output either $Q \cap \mathit{CC}(\Pi)$ (i.e., the largest subset of $Q$ that only contains cautious consequences of $\Pi$) in case $\Pi$ is coherent, or $\bot$ when $\Pi$ is incoherent.
To this aim, an underestimation $U$ and an overestimation $O$ are initially set to $\emptyset$ and $Q$, respectively, and updated during computation.
A first coherence test of $\Pi$ is performed by means of function $\mathrm{ComputeStableModel}$ (line~2), so that $\bot$ is returned if $\Pi$ is incoherent (line~3).
Otherwise, if a stable model $I$ is found, it is used to possibly reduce the overestimation (line~4), as atoms not in $I$ cannot be cautious consequences of $\Pi$.

\begin{example}[Execution of Algorithm~\ref{alg:cautious}]\label{ex:firstmodel}
Consider program $\Pi_\mathit{run}$ from Example~\ref{ex:program}, and the execution of $\mathrm{CautiousReasoning}$$(\textsc{refine},$ $\Pi_\mathit{run}, \{q_1, q_2,$ $q_3, q_4\})$.
The call to $\mathrm{ComputeStableModel}(\Pi_\mathit{run}, \emptyset)$ at line~2 returns a stable model of $\Pi_\mathit{run}$, say $I_1$.
Hence, $O$ is updated to $\{q_1, q_3, q_4\}$, i.e., the candidate $q_2$ is discarded.
$\hfill \lhd$	
\end{example}

After the  initial steps, the underestimation and the overestimation are refined until they are equal, both representing exactly the cautious consequences of $\Pi$ (lines~5--6).
The way the estimations are iteratively refined depends on the cautious reasoning algorithm (or strategy) at hand.
Two previously proposed and implemented strategies, referred to as \emph{overestimate reduction} (\textsc{or}) and \emph{iterative coherence testing} (\textsc{ict}), are detailed next.

\begin{procedure}[t]
	\caption{\textsc{or}()}
	$(I, C) := \mathrm{ComputeStableModel}(\Pi \cup \{\bot \leftarrow O\}, \emptyset)$\;
	\lIf{$I = \bot$}{
		$U := O$\;
	}
	\lElse{
		$O := O \cap I$\;
	}
\end{procedure}
Strategy \textsc{or} iteratively searches for new stable models to improve the overestimation $O$.
To this aim, a constraint of the form $\bot \leftarrow O$ is added to $\Pi$, so that at least one of the atoms in $O$ must be false in the computed stable model.
If such a stable model does not exist, then the underestimation is set equal to the overestimation to terminate the computation.

\begin{example}[Execution of \textsc{or}; continuing Example~\ref{ex:firstmodel}]
A stable model of $\Pi_\mathit{run} \cup \{\bot \leftarrow q_1, q_3, q_4\}$ is searched for.
Say that $I_4$ is found. Then $O$ is set to $\{q_1, q_3\}$, and a stable model of $\Pi_\mathit{run} \cup \{\bot \leftarrow q_1, q_3\}$ is searched for.
However, the program is incoherent, and therefore $U$ is set to $\{q_1, q_3\}$, and the algorithm terminates. $\hfill \lhd$	
\end{example}

\begin{procedure}[t]
	\caption{\textsc{ict}()}
	$a := \mathrm{OneOf}(O \setminus U)$\;
	$(I, C) := \mathrm{ComputeStableModel}(\Pi, \{\naf a\})$\;
	\lIf{$I = \bot$}{
		$U$ := $U \cup \{a\}$\;
	}
	\lElse{
		$O$ := $O \cap I$\;
	}
\end{procedure}
Strategy \textsc{ict} can improve both estimations during its computation.
In fact, one cautious consequence candidate, say $a$, is selected by a heuristic function $\mathrm{OneOf}$.
The complement of the selected candidate is put in the assumption literals in order to search for a stable model of $\Pi$ not containing $a$.
If such a stable model does not exist, then the underestimation can be improved by adding $a$.
Otherwise, the stable model found is used to improve the overestimation.

\begin{example}[Execution of \textsc{ict}; continuing Example~\ref{ex:firstmodel}]
Say that $\mathrm{OneOf}(\{q_1,q_3,q_4\})$ selects $q_1$.
$\mathrm{ComputeStableModel}(\Pi_\mathit{run}, \{\naf q_1\})$ returns $(\bot,\{\naf q_1\})$ because there are no stable models of $\Pi_\mathit{run}$ where $q_1$ is false.
Hence $q_1$ is added to $U$.
After that, say that $\mathrm{OneOf}(\{q_3,q_4\})$ selects $q_4$.
$\mathrm{ComputeStableModel}(\Pi_\mathit{run}, \{\naf q_4\})$ returns a stable model, say $I_4 = \{b, c, q_1, q_2, q_3\}$.
Hence, $O$ is set to $\{q_1, q_3\}$, and $\mathrm{OneOf}(\{q_3\})$ has to return $q_3$.
Finally, $\mathrm{ComputeStableModel}(\Pi_\mathit{run}, \{\naf q_3\})$ returns $(\bot, \{\naf q_3\})$, and therefore $q_3$ is added to $U$, and the algorithm terminates returning $\{q_1, q_3\}$. $\hfill \lhd$
\end{example}

\section{New algorithms for cautious reasoning in ASP}
\label{sec:newalgo}

The skeleton provided by Algorithm~\ref{alg:cautious} can be used by strategies other than \textsc{ict} and \textsc{or}, as long as the estimations are iteratively updated and are guaranteed to converge to the set of cautious consequences of the input program $\Pi$ among the candidates in $Q$.
In this section, new strategies are presented, with the aim of discarding several candidates at once and thereby speeding up the convergence of the estimations.
Specifically, one of the new strategies is based on the search of minimal stable models, and actually results into a meta-algorithm that can be instantiated with different search procedures (Section~\ref{sec:minimal}).
The other strategy introduced here is based on an alternative characterization of cautious consequences in terms of the notion of unsatisfiable cores used in this paper~(Section~\ref{sec:cm}).

\subsection{Cautious reasoning via minimal models}\label{sec:minimal}

\begin{procedure}[t]
	\caption{\textsc{min}()}
	$I := \mathrm{ComputeMinStableModel}(\Pi, O)$\;
	\lIf{$I \supseteq O$} {
		$U := O$\;
	}
	\lElse {
		$O := O \cap I$\;
	}
\end{procedure}
A stable model $I$ of a program $\Pi$ is said to be \emph{minimal} with respect to a set $O$ of objective atoms if there is no $I' \in \mathit{SM}(\Pi)$ such that $I' \cap O \subset I \cap O$.
Strategy \textsc{min} takes advantage of this notion, specifically of function $\mathrm{ComputeMinStableModel}$, whose input are a coherent program $\Pi$ and a set $O$ of objective atoms, and whose output is a minimal stable model of $\Pi$.
Armed with such a function, strategy \textsc{min} searches for stable models of $\Pi$ that are minimal with respect to the current overestimation $O$.
The stable models returned by $\mathrm{ComputeMinStableModel}$ either discard some candidate from $O$, which would lead to a new iteration of the strategy, or are such that all atoms in $O$ are true.
In the latter case, the set $O$ only contains cautious consequences of $\Pi$, and therefore the underestimation is updated to terminate the algorithm.

\begin{example}[Execution of \textsc{min}; continuing Example~\ref{ex:firstmodel}]
Say that $\mathrm{ComputeMinStableModel}(\Pi_\mathit{run}, \{q_1, q_3, q_4\})$ returns $I_2 = \{a,c,q_1,q_2,q_3\}$.
Hence $O$ is set to $\{q_1, q_3\}$.
Say that the next call to $\mathrm{ComputeMinStableModel}(\Pi_\mathit{run}, \{q_1, q_3\})$ returns $I_3 = \{b, d, q_1, q_3, q_4\}$.
Since $I_3 \supseteq O$, $U$ is set to $\{q_1, q_3\}$, which is returned by the algorithm. $\hfill \lhd$	
\end{example}

\begin{restatable}{theorem}{ThmCorr}
Given a program $\Pi$, and a set $Q$ of atoms, $\mathrm{CautiousReasoning}(\textsc{min},\Pi,Q)$ terminates and returns either $\bot$ (if $\mathit{SM}(\Pi) = \emptyset$) or 
$\mathit{CC}(\Pi) \cap Q$ (otherwise).
\end{restatable}

\begin{proof}
	Let $\Pi$ be coherent (otherwise correctness is immediate).
	Termination is guaranteed by the fact that at each iteration a stable model $I$ is found, and either $O$ is shrunk (if $I \not\supseteq O$), or $U$ is set equal to $O$.
	In fact, if $O \neq \mathit{CC}(\Pi) \cap Q$, then there is a stable model $I$ of $\Pi$ such that $I \not\supseteq O$, and therefore any stable model $I'$ of $\Pi$ such that $I' \supseteq O$ is not minimal with respect to $O$;
	hence, in this case $O$ is shrunk.
	Otherwise, if $O = \mathit{CC}(\Pi) \cap Q$, then every minimal stable model $I$ of $\Pi$ with respect to $O$ must be such that $I \supseteq O$, and therefore $O$ is the answer to be provided in output.
	\hfill
\end{proof}

\textsc{min} is actually a meta-algorithm that can be instantiated in several ways by using different versions of function $\mathrm{ComputeMinStableModel}$.
Two versions of this function are presented, one based on an algorithm developed in the context of qualitative preferences (\textsc{opt}; \citeNP{DBLP:journals/constraints/RosaGM10}) and another based on a MaxSAT algorithm (\textsc{one}; \citeNP{DBLP:conf/ijcai/AlvianoDR15}).

\begin{function}[t]
	\caption{\textsc{opt}($\Pi$, $O$)}\label{alg:opt}
    
    \SetKwFor{Loop}{loop}{}{}
	
	$A := [\,]$\tcp*{stack of assumptions, also used as a set}
	\Loop{}{
	    \While(\tcp*[f]{unassigned candidates?}){$O \setminus (A \cup \overline{A}) \neq \emptyset$}{
	        $A.\mathit{push}(\naf \mathit{OneOf}(O \setminus (A \cup \overline{A})))$\;
	        $\mathit{propagate}(\Pi,A)$\;
	    }
    	$(I,C) := \mathit{ComputeStableModel}(\Pi, A)$\;
    	\lIf(\tcp*[f]{$I$ is a minimal model}){$I \neq \bot$}{\Return $I$}
		\lWhile(\tcp*[f]{restore consistency}){$C \not\subseteq A$}{$A.\mathit{pop}()$}
	}
\end{function}
\medskip\noindent\textsc{opt}.
Subset-minimality is among the qualitative preferences analyzed by \citeANP{DBLP:journals/constraints/RosaGM10}~\citeyear{DBLP:journals/constraints/RosaGM10}, and addressed by modifying the search procedure of a SAT solver;
the resulting algorithm is referred to as \textsc{opt} in the literature, and was also adapted to ASP \cite{Alviano01092015,DBLP:conf/aaai/GebserKROSW13}.
Intuitively, the branching heuristic of the solver is forced to select $\overline{p}$ for $p \in O$, where $O$ is the set of objective atoms, before any other unassigned literal.
In this way, the search is driven to falsify as many atoms in $O$ as possible.
When all atoms in $O$ are assigned, standard stable model search procedure is applied without subjecting the branching heuristic to any further guidance.
Hence, if a stable model is found, it is guaranteed to be minimal with respect to the set of objective atoms \cite{DBLP:journals/constraints/RosaGM10}.
If instead the current assignment to atoms in $O$ cannot be extended to a stable model, then a conflict involving some objective atom has to be detected.
In this case, the assignment of some atom in $O$ is flipped because of the constraint learned by analyzing the conflict, and hence the procedure is repeated with a different assignment for the objective atoms.
Function \textsc{opt} reports such a strategy.
A stack $A$ of assumption literals, initially empty (line~1), is populated with complements of candidates in $O$ (line~4).
After each insertion, $\mathit{propagate}(\Pi,A)$ is used to extend $A$ with (unit) implied literals, if possible;
otherwise, in case of conflict, $\Pi$ is extended with a learned clause, some literal in $A$ is flipped, and propagation is repeated.
When all atoms in $O$ occur in $A$, a stable model extending $A$ is searched (line~6).
If it is found, it is returned (line~7).
Otherwise, some literals in $A$ are removed so to not incur again in the returned unsatisfiable core (line~8).

\begin{example}[Execution of \textsc{opt}; continuing Example~\ref{ex:firstmodel}]
Say that $\textsc{opt}(\Pi_\mathit{run},\{q_1,q_3,q_4\})$ starts by selecting $\naf q_1$.
Hence $\naf a$ and $\naf b$ are inferred via $r_3$ and $r_4$, respectively.
However, this is not possible due to $r_1$, that is, there is a conflict.
The program is therefore extended with the learned constraint $\bot \leftarrow \naf q_1$, so that $~q_1$ cannot be selected anymore by the branching heuristic.
Now, say that $\naf q_4$ is selected.
Literal $\naf d$ is inferred via $r_8$, and therefore $c$ and $q_3$ are inferred via $r_2$ and $r_7$, respectively.
(In addition, $q_2$ is then inferred via $r_5$.)
After that, the branching heuristic is not subject to any guidance, and a stable model of $\Pi_\mathit{run}$, say $I_2 = \{a,c,q_1,q_2,q_3\}$, is returned.
$\hfill \lhd$
\end{example}

\begin{function}[t]
	\caption{\textsc{one}($\Pi$, $O$)}\label{alg:one}
    
    \SetKwFor{Loop}{loop}{}{}
	
	$S := \overline{O}$\;
	\Loop{}{
    	$(I,C) := \mathit{ComputeStableModel}(\Pi, S)$\;\label{alg:one:ln:solve}
    	\lIf{$I \neq \bot$}{\Return $I$\;}
		Let $C$ be $\{\ell_0,\ldots, \ell_n\}$ (for some $n \geq 0$), and $p_1,\ldots,p_n$ be $|C|-1$ fresh atoms\;\label{alg:one:ln:analyze}
		$S := (S \setminus C) \cup \{\naf p_1, \ldots, \naf p_n\}$\;
		$\Pi := \Pi \cup \{\{\ell_0,\ldots,\ell_n, p_1,\ldots,p_n\} \geq n\} \cup \{\bot \leftarrow p_i, \naf p_{i-1} \mid i \in [2..n]\}$\;
	}
\end{function}
\noindent{\textsc{one}.}
MaxSAT algorithms, being the state-of-the-art for the computation of cardinality-minimal models, can be also used to compute subset-minimal models. Specifically,
the  algorithm \textsc{one} has been used to enumerate models of circumscribed theory \cite{DBLP:journals/tplp/Alviano17}.
Simplified to the setting of this paper, function \textsc{one} takes as input a coherent program $\Pi$ and a set $O$ of objective atoms, and returns a minimal stable model of $\Pi$ with respect to $O$.
Hence, it is a valid instantiation of $\mathrm{ComputeMinStableModel}$.
In more detail, \textsc{one} keeps a set $S$ of \emph{soft literals} to be maximized, initially set to the complements of the atoms in $O$ (line~1).
At each step of computation, a stable model of $\Pi$ subject to the assumptions $S$ is searched (line~3), and eventually returned (line~4).
When the search fails, instead, soft literals in the computed unsatisfiable core are replaced by fresh literals (line~6), and a choice rule enforcing the satisfaction of at least $n$ literals among those in the unsatisfiable core and those fresh is added to the program (line~7).
Since the next stable model search is forced to assign false to the fresh atoms, the choice rule actually drives the search to a stable model containing at least $n$ literals among those in the unsatisfiable core.
Additionally, constraints of the form $\bot \leftarrow p_i, \naf p_{i-1}$ are added to the program to eliminate symmetric solutions.

\begin{example}[Execution of \textsc{one}; continuing Example~\ref{ex:firstmodel}]
$\textsc{one}(\Pi_\mathit{run},\{q_1,q_3,q_4\})$ initially searches for a stable model of $\Pi_\mathit{run}$ subject to the assumption literals $\{\naf q_1, \naf q_3, \naf q_4\}$.
Say that $(\bot,\{\naf q_1,\naf q_4\})$ is returned.
$\Pi_\mathit{run}$ is extended with $\{\naf q_1, \naf q_4, p_1\} \geq 1$,
and $S$ is modified to $\{\naf q_3, \naf p_1\}$.
Say that the next stable model search returns $(\bot,\{\naf q_3\})$, so that $S$ is now $\{\naf p_1\}$.
The next stable model search returns a model, say $I_2 = \{a,c,q_1,q_2,q_3\}$.
$\hfill \lhd$
\end{example}

\subsection{Cautious reasoning via unsatisfiable core minimization}\label{sec:cm}

Unsatisfiable cores, as considered in this paper, provide an alternative characterization of the notion of cautious consequences,  formalized as
the following proposition.

\begin{restatable}{proposition}{PropCores}\label{th:cores}
Assume that $\Pi$ is a program and $A$ is a set of atoms.
Then $p \in \mathit{CC}(\Pi) \cap A$ if and only if $\{\naf p\}$ is an unsatisfiable core of $\Pi$ with respect to $A$.
\end{restatable}

\begin{proof}
	Follows from $p \in \mathit{CC}(\Pi) \cap A$ if and only if
	$p \in I$ for all $I \in \mathit{SM}(\Pi)$ if and only if
	$\mathit{SM}(\Pi \cup \{\bot \leftarrow \naf p\}) = \emptyset$ if and only if
	$\{\naf p\}$ is an unsatisfiable core of $\Pi$ with respect to $A$.
	Note that the claim also holds for incoherent programs:
	$\mathit{SM}(\Pi) = \emptyset$ implies $\mathit{CC}(\Pi) = \mathcal{A}$, and $\mathit{SM}(\Pi \cup \{\bot \leftarrow \naf p\}) = \emptyset$ for all $p \in \mathcal{A}$.
	\hfill
\end{proof}

An immediate consequence of Proposition~\ref{th:cores} is that a subset-minimal unsatisfiable core containing more than one atom can be used to discard candidates, since
none of the atoms in the unsatisfiable core are in fact cautious consequences of the input program in this case.
Algorithm \textsc{cm} takes advantage of this fact, and searches for subset-minimal unsatisfiable cores.
If a singleton is found, a cautious consequence of the input program is identified.
Otherwise, all atoms in the unsatisfiable core are discarded.
Actually, in order to compute a subset-minimal unsatisfiable core, an initial unsatisfiable core is minimized by searching for stable models containing all atoms in the unsatisfiable core but one.
It turns out that such a stable model can be already used to discard several candidates, even before the minimization process is completed.

\begin{procedure}[t]
	\caption{\textsc{cm}()}
	$C := \overline{O \setminus U}$;\qquad $C' := \emptyset$\;
	\While{$C \neq \emptyset$}{
		$(I,C) := \mathrm{ComputeStableModel}(\Pi, C)$\;
		\uIf{$I \neq \bot$}{$O := O \cap I$;\qquad $C := C'$;\qquad $C' := \emptyset$\;}
		\Else{
			$C' := \{\mathrm{OneOf}(C)\}$;\qquad $C := C \setminus C'$\;
		}
	}
	\lIf{$C' \neq \emptyset$}{$U := U \cup \overline{C'}$\;}
\end{procedure}
In more detail, \textsc{cm} stores in $C$ the unsatisfiable core candidate, that is, the assumption literals for stable model searches performed by the algorithm.
Initially, $C$ contains the complements of all unknown candidates (line~1), and a stable model is searched (line~3).
If the search fails, then the unsatisfiable core identified by $\mathrm{ComputeStableModel}$ is partitioned among $C$ and $C'$, where $C'$ only contains one heuristically selected literal (line~7).
This process is repeated until either $C$ is empty, or a stable model is found.
A stable model would necessarily assign all literals in $C$ as true, so that the associated candidates can be discarded at once (line~5);
in this case the only candidate possibly involved in a singleton unsatisfiable core is the one stored in $C'$, if any, and therefore it is moved into $C$ (line~5; if it is involved in an unsatisfiable core, it will be moved back to C' at line~7).
When the while-loop terminates, if $C'$ still contains a single candidate, it is guaranteed to be a singleton unsatisfiable core, and therefore a cautious consequence of the input program (line~8).

\begin{example}[Execution of \textsc{cm}; continuing Example~\ref{ex:firstmodel}]
Initially $C = \{\naf q_1, \naf q_3, \naf q_4\}$ and $C' = \emptyset$.
Say that $\mathrm{ComputeStableModel}(\Pi_\mathit{run}, C)$ returns $(\bot,$ $\{\naf q_1, \naf q_4\})$, and that $\mathrm{OneOf}(\{\naf q_1, \naf q_4\})$ returns $\naf q_1$, so that the next iteration has $C' = \{\naf q_1\}$ and $C = \{\naf q_4\}$.
A stable model is found, say $I_2 = \{a,c,q_1,q_2,$ $q_3\}$, and $O$ is shrunk to $\{q_1,q_3\}$.
Moreover, the next iteration has $C = \{\naf q_1\}$ and $C' = \emptyset$.
Hence, $(\bot, \{\naf q_1\})$ is returned by \linebreak $\mathrm{ComputeStableModel}$, and $q_1$ is added to $U$ (after moving its complement into $C'$).
$\hfill \lhd$	
\end{example}

\begin{restatable}{theorem}{ThmCM}
Given a program $\Pi$, and a set $Q$ of atoms, $\mathrm{CautiousReasoning}(\textsc{cm},\Pi,Q)$ terminates, and returns either $\bot$ (if $\mathit{SM}(\Pi) = \emptyset$) or 
 $\mathit{CC}(\Pi) \cap Q$ (otherwise).
\end{restatable}

\begin{proof}
	Let $\Pi$ be coherent (otherwise correctness is immediate).
	First of all, any execution of procedure \textsc{cm} terminates.
	In fact, if $\overline{O \setminus U}$ is not an unsatisfiable core of $\Pi$, then the stable model returned by $\mathrm{ComputeStableModel}$ is used to discard all unknown candidates, and $C$ is assigned the empty set.
	Otherwise, if $\overline{O \setminus U}$ is an unsatisfiable core of $\Pi$, $\mathrm{ComputeStableModel}$ returns a possibly smaller unsatisfiable core, which is partitioned among $C$ and the singleton $C'$;
	the process is repeated ($|O| - |U|$ times in the worst case) until $C$ is not an unsatisfiable core, and therefore possibly used to discard candidates from $O$.
	At this point, if $C'$ is an unsatisfiable core, the complement of the literal in $C'$ is a cautious consequence of $\Pi$ because of Proposition~\ref{th:cores}.
	Second, the algorithm itself terminates because each execution of procedure \textsc{cm} shrinks $O$ or extends $U$ (or both);
	hence, \textsc{cm} is executed at most $|O| - |U|$ times.
	Finally, note that no cautious consequence of $\Pi$ can be eliminated at line~5 because $\mathit{CC}(\Pi) \subseteq I$ holds for all $I \in \mathit{SM}(\Pi)$, and therefore all of them are added to $U$ at line~8 during some execution of procedure \textsc{cm}.
	\hfill
\end{proof}

\section{Experiments}
\label{sec:exp}

The new algorithms are implemented in \wasp \cite{DBLP:journals/tplp/AlvianoDR14}, and their performance was measured on all instances from the latest ASP Competitions \cite{DBLP:journals/tplp/CalimeriIR14,DBLP:journals/ai/CalimeriGMR16,DBLP:journals/jair/GebserMR17} involving non-ground queries.
(Ground queries are not included in the experiment because they are usually handled by adding an integrity constraint to the input program, and by searching for a single answer set.)
Instances are grounded by \textsc{gringo} v. 4.4.0 \cite{DBLP:conf/lpnmr/GebserKKS11}, which therefore also produces the set of candidates.
As a reference to the state of the art, \clasp v. 3.3.3 \cite{DBLP:journals/ai/GebserKS12} was also tested;
\clasp implements algorithm \textsc{or}.
The DLV solver is not considered here as its performance on cautious reasoning has been shown in earlier work to be dominated by
the other approaches considered here~\cite{DBLP:journals/tplp/AlvianoDR14}.
The experiments were run on computing nodes with Intel Xeon 2.4-GHz processors and 16 GB of memory.
Per-instance time and memory limits were set to 600 seconds and 15 GB, respectively.
In the following, \textsc{wasp-}$X$ refers to \wasp executing the strategy $X \in \{\textsc{or},$ $\textsc{ict}$, $\textsc{opt},$ $\textsc{one},$ $\textsc{cm}\}$.

\begin{figure}[t]
	\figrule
	\begin{tikzpicture}[scale=0.85]
	\pgfkeys{%
		/pgf/number format/set thousands separator = {}}
	\begin{axis}[
	scale only axis
	, font=\normalsize
	, x label style = {at={(axis description cs:0.5,0.0)}}
	, y label style = {at={(axis description cs:0.0,0.5)}}
	, xlabel={Number of solved instances}
	, ylabel={Per-instance time limit (s)}
	, xmin=18, xmax=200
	, ymin=0, ymax=600
	, legend style={at={(0.12,0.96)},anchor=north, draw=none,fill=none}
	, legend columns=1
	, width=1\textwidth
	, height=0.5\textwidth
	, ytick={0,150,300,450,600}
	, major tick length=2pt
	]

        \addplot [mark size=3pt, color=blue, mark=star] [unbounded coords=jump] table[col sep=semicolon, y index=2] {./cactusnoideal.csv};
        \addlegendentry{\waspict}
	\addplot [mark size=3pt, color=green!50!black, mark=triangle] [unbounded coords=jump] table[col sep=semicolon, y index=3] {./cactusnoideal.csv}; 
	\addlegendentry{\waspor}
	
        \addplot [mark size=3pt, color=orange, mark=o] [unbounded coords=jump] table[col sep=semicolon, y index=5] {./cactusnoideal.csv};
        \addlegendentry{\waspopt}

        \addplot [mark size=3pt, color=black, mark=x] [unbounded coords=jump] table[col sep=semicolon, y index=4] {./cactusnoideal.csv};
        \addlegendentry{\waspone}

        \addplot [mark size=3pt, color=yellow!50!black, mark=square] [unbounded coords=jump] table[col sep=semicolon, y index=6] {./cactusnoideal.csv};
        \addlegendentry{\waspcm}

	\addplot [mark size=3pt, color=red, mark=pentagon] [unbounded coords=jump] table[col sep=semicolon, y index=1] {./cactusnoideal.csv}; 
	\addlegendentry{\clasp}

	\end{axis}
	\end{tikzpicture}

	\caption{Performance comparison on non-ground queries in ASP Competitions: number of solved instances within a given per-instance time limit.}\label{fig:cactus}
	\figrule
\end{figure}
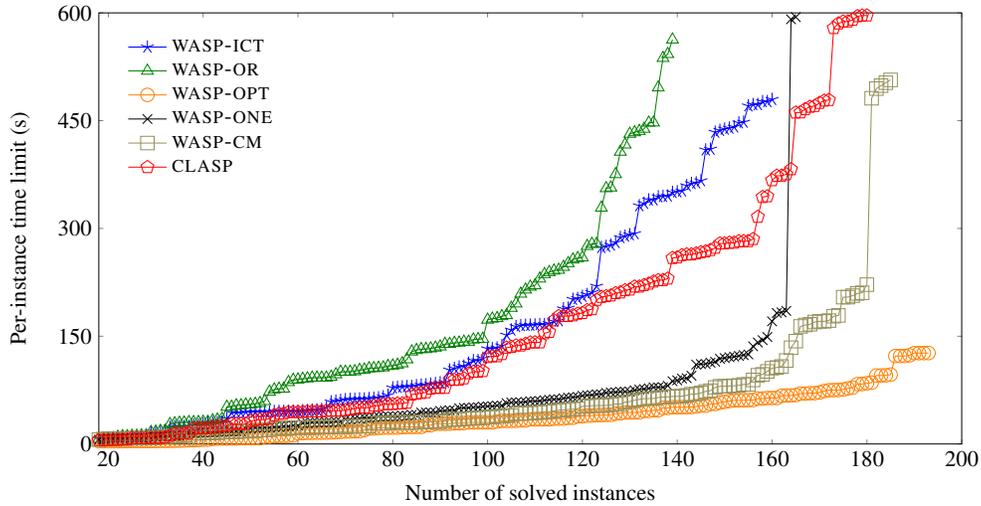
Results are reported in the cactus plot of Figure~\ref{fig:cactus}, showing for each algorithm
 the number of instances solved (x-axis) under different per-instance time limits (y-axis).
As a first observation, \wasp cannot reach the performance of \clasp on the execution of algorithm \textsc{or}, and indeed \clasp solved 41 instances more than \waspor.
However, such a huge gap is completely filled by \waspopt and \waspcm, which actually solve 13 and 5 instances more than \clasp, respectively.
The best performance overall is obtained by \waspopt, which solves all the instances with an average running time of 35.8 seconds.
Interestingly, \waspopt is able to solve 95\% of the tested instances within two minutes, and anyhow all test cases within 150 seconds of computation.

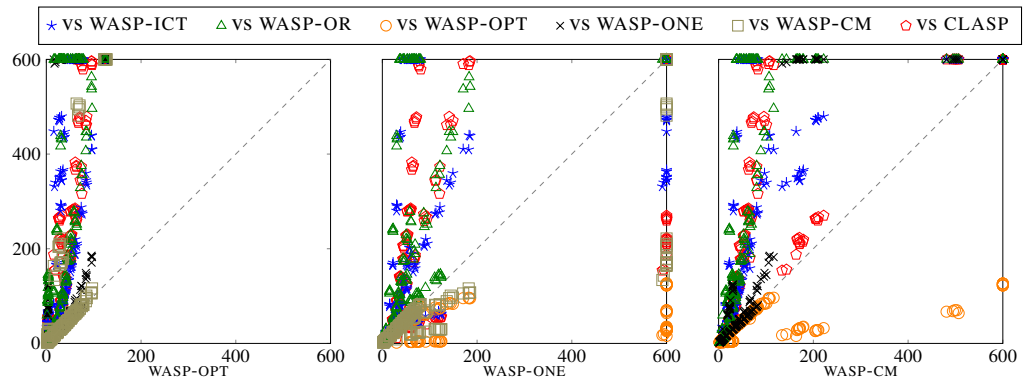
\begin{figure}
    \figrule
    \,\hfill
    \begin{tikzpicture}
        \begin{axis}[
            hide axis,
            width=0.6\textwidth,
            height=0.12\textwidth,
            legend style={at={(0.5,1.3)}, anchor=north, align=left, /tikz/every even column/.append style={column sep=1em}},
            legend columns=-1,
            xmin=0,
            xmax=600,
            ymin=0,
            ymax=600,
        ]      
        \addlegendimage{mark size=2pt, only marks, color=blue, mark=star};
        \addlegendentry{vs \waspict}            
        \addlegendimage{mark size=2pt, only marks, color=green!50!black, mark=triangle};
        \addlegendentry{vs \waspor}
        \addlegendimage{mark size=2pt, only marks, color=orange, mark=o};
        \addlegendentry{vs \waspopt}
        \addlegendimage{mark size=2pt, only marks, color=black, mark=x};
        \addlegendentry{vs \waspone}
        \addlegendimage{mark size=2pt, only marks, color=yellow!50!black, mark=square};
		\addlegendentry{vs \waspcm}            
        \addlegendimage{mark size=2pt, only marks, color=red, mark=pentagon};
        \addlegendentry{vs \clasp}
        \end{axis}
    \end{tikzpicture}%

    \begin{tikzpicture}[scale=0.7]
        \pgfkeys{%
            /pgf/number format/set thousands separator = {}}
        \begin{axis}[
        scale only axis
        , font=\normalsize
        , x label style = {at={(axis description cs:0.5,0.04)}}
        , y label style = {at={(axis description cs:0.05,0.5)}}
        , xlabel={\waspopt}
        , width=.4\textwidth
        , height=.4\textwidth
        , xmin=0, xmax=600
        , ymin=0, ymax=600
        , xtick={0,200,400,600}
        , ytick={0,200,400,600}
        , major tick length=2pt
        ]
        \addplot [mark size=3pt, only marks, color=red, mark=pentagon] [unbounded coords=jump] table[col sep=semicolon, x index=5, y index=1] {./scatter.csv}; 	
        \addlegendentry{vs \clasp}
        \addplot [mark size=3pt, only marks, color=blue, mark=star] [unbounded coords=jump] table[col sep=semicolon, x index=5, y index=2] {./scatter.csv}; 	
        \addlegendentry{vs \waspict}
        \addplot [mark size=3pt, only marks, color=black, mark=x] [unbounded coords=jump] table[col sep=semicolon, x index=5, y index=3] {./scatter.csv}; 
        \addlegendentry{vs \waspone}
        \addplot [mark size=3pt, only marks, color=green!50!black, mark=triangle] [unbounded coords=jump] table[col sep=semicolon, x index=5, y index=4] {./scatter.csv}; 	
        \addlegendentry{vs \waspor}
        \addplot [mark size=3pt, only marks, color=yellow!50!black, mark=square] [unbounded coords=jump] table[col sep=semicolon, x index=5, y index=6] {./scatter.csv}; 	
        \addlegendentry{vs \waspcm}
        \addplot [color=gray, dashed] [unbounded coords=jump] table[col sep=semicolon, x index=0, y index=0] {./scatter.csv}; 
        
        \legend{}
        \end{axis}
    \end{tikzpicture}
    \hfill
    \begin{tikzpicture}[scale=0.7]
        \pgfkeys{%
            /pgf/number format/set thousands separator = {}}
        \begin{axis}[
        scale only axis
        , font=\normalsize
        , x label style = {at={(axis description cs:0.5,0.04)}}
        , y label style = {at={(axis description cs:0.05,0.5)}}
        , xlabel={\waspone}
        , width=.4\textwidth
        , height=.4\textwidth
        , xmin=0, xmax=600
        , ymin=0, ymax=600
        , xtick={0,200,400,600}
        , ytick={-1}
        , major tick length=2pt
        ]
        \addplot [mark size=3pt, only marks, color=red, mark=pentagon] [unbounded coords=jump] table[col sep=semicolon, x index=3, y index=1] {./scatter.csv}; 	
        \addlegendentry{vs \clasp}
        \addplot [mark size=3pt, only marks, color=blue, mark=star] [unbounded coords=jump] table[col sep=semicolon, x index=3, y index=2] {./scatter.csv}; 	
        \addlegendentry{vs \waspict}
        \addplot [mark size=3pt, only marks, color=orange, mark=o] [unbounded coords=jump] table[col sep=semicolon, x index=3, y index=5] {./scatter.csv}; 
        \addlegendentry{vs \waspopt}
        \addplot [mark size=3pt, only marks, color=green!50!black, mark=triangle] [unbounded coords=jump] table[col sep=semicolon, x index=3, y index=4] {./scatter.csv}; 	
        \addlegendentry{vs \waspor}
        \addplot [mark size=3pt, only marks, color=yellow!50!black, mark=square] [unbounded coords=jump] table[col sep=semicolon, x index=3, y index=6] {./scatter.csv}; 	
        \addlegendentry{vs \waspcm}
        \addplot [color=gray, dashed] [unbounded coords=jump] table[col sep=semicolon, x index=0, y index=0] {./scatter.csv}; 
        
        \legend{}
        \end{axis}
    \end{tikzpicture}
    \hfill
    \begin{tikzpicture}[scale=0.7]
        \pgfkeys{%
            /pgf/number format/set thousands separator = {}}
        \begin{axis}[
        scale only axis
        , font=\normalsize
        , x label style = {at={(axis description cs:0.5,0.04)}}
        , y label style = {at={(axis description cs:0.05,0.5)}}
        , xlabel={\waspcm}
        , width=.4\textwidth
        , height=.4\textwidth
        , xmin=0, xmax=600
        , ymin=0, ymax=600
        , xtick={0,200,400,600}
        , ytick={-1}
        , major tick length=2pt
        ]
        \addplot [mark size=3pt, only marks, color=red, mark=pentagon] [unbounded coords=jump] table[col sep=semicolon, x index=6, y index=1] {./scatter.csv}; 	
        \addlegendentry{vs \clasp}
        \addplot [mark size=3pt, only marks, color=blue, mark=star] [unbounded coords=jump] table[col sep=semicolon, x index=6, y index=2] {./scatter.csv}; 	
        \addlegendentry{vs \waspict}
        \addplot [mark size=3pt, only marks, color=orange, mark=o] [unbounded coords=jump] table[col sep=semicolon, x index=6, y index=5] {./scatter.csv}; 
        \addlegendentry{vs \waspopt}
        \addplot [mark size=3pt, only marks, color=green!50!black, mark=triangle] [unbounded coords=jump] table[col sep=semicolon, x index=6, y index=4] {./scatter.csv}; 	
        \addlegendentry{vs \waspor}
        \addplot [mark size=3pt, only marks, color=black, mark=x] [unbounded coords=jump] table[col sep=semicolon, x index=6, y index=3] {./scatter.csv}; 	
        \addlegendentry{vs \waspone}
        \addplot [color=gray, dashed] [unbounded coords=jump] table[col sep=semicolon, x index=0, y index=0] {./scatter.csv}; 
        
        \legend{}
        \end{axis}
    \end{tikzpicture}
    
    \caption{Instance-by-instance comparison on non-ground queries in ASP Competitions.}\label{fig:scatter}
    \figrule
\end{figure}
An instance-by-instance comparison is reported on Figure~\ref{fig:scatter}.
The scatter plots clearly show that the advantage of \textsc{wasp-opt} and \textsc{wasp-cm} is uniform on all test cases, while the result is mixed for \textsc{wasp-one}.
For the tested instances, \waspopt is faster than \waspone because computing cardinality-minimal models is usually harder than computing subset-minimal models;
\waspopt is faster than \waspcm because \waspcm performs some hard stable model searches due to the unsatisfiable core minimization.

\begin{table}[b!]
	\caption{
		Numbers of solved instances and cumulative running time (in seconds; each timeout adds 600 seconds) on the benchmarks with non-ground queries from ASP Competitions.
		Best performance emphasized in bold (within this respect, differences up to 10 seconds are ignored).
	}
	\label{tab:queryanswering}
	\centering
	\tabcolsep=0.040cm
	\begin{tabular}{lrrrrrrrrrrrrr}
		\toprule
		& & \multicolumn{2}{c}{\textbf{\waspor}}	&\multicolumn{2}{c}{\textbf{\waspict}} &\multicolumn{2}{c}{\textbf{\waspopt}} &\multicolumn{2}{c}{\textbf{\waspone}} &\multicolumn{2}{c}{\textbf{\waspcm}} &\multicolumn{2}{c}{\textbf{\clasp}}\\
		\cmidrule{3-4}\cmidrule{5-6}\cmidrule{7-8}\cmidrule{9-10}\cmidrule{11-12}\cmidrule{13-14}
		\textbf{Benchmark} &  \textbf{\#} &	\textbf{sol}.	&	\textbf{sum t}	&	\textbf{sol.}	&	\textbf{sum t}	&	\textbf{sol.}	&	\textbf{sum t}	&	\textbf{sol.}	&	\textbf{sum t}	&	\textbf{sol.}	&	\textbf{sum t}	&	\textbf{sol.}	&	\textbf{sum t}\\
		\cmidrule{1-14}
		CQA-Q3	&	40 & 40	&	4346.7	&	40	&	3384.9	&	40	&	{1276.3}	&	40	&	\textbf{1272.5}	&	40	&	1312.8	&	40	&	4353.9\\
		CQA-Q6	&	40 & 40	&	8410.5	&	40	&	6799.5	&	40	&	\textbf{1956.1}	&	40	&	2948.5	&	40	&	2149.1	&	40	&	8505.0\\
		CQA-Q7	&	40 & 16	&	17636.9	&	20	&	15999.3	&	40	&	\textbf{1680.9}	&	40	&	1701.1	&	40	&	1740.8	&	40	&	8929.5\\
		MCSQ	&	73 & 43	&	20646.6	&	60	&	15875.5	&	\textbf{73}	&	\textbf{1995.3}	&	45	&	20050.1	&	65	&	11006.5	&	60	&	12701.0\\
		\cmidrule{1-14}
		\textbf{Total}	&	193	& 139	&	51040.7	&	160	&	42059.2	&	\textbf{193}	&	\textbf{6908.5}	&	165	&	25972.2	&	185	&	16209.3	&	180	&	34489.4\\
		\bottomrule
	\end{tabular}
\end{table}
Details on solved instances and the cumulative running times for benchmarks including non-ground queries from ASP Competitions are given in Table~\ref{tab:queryanswering}.
As already observed before, the best performance is obtained by \textsc{wasp-opt} and \textsc{wasp-cm}.

\begin{figure}[t]
	\figrule
	\begin{tikzpicture}[scale=0.85]
	\pgfkeys{%
		/pgf/number format/set thousands separator = {}}
	\begin{axis}[
	scale only axis
	, font=\normalsize
	, x label style = {at={(axis description cs:0.5,0.0)}}
	, y label style = {at={(axis description cs:0.0,0.5)}}
	, xlabel={Number of solved instances}
	, ylabel={Per-instance time limit (s)}
	, xmin=0, xmax=55
	, ymin=0, ymax=620
	, legend style={at={(0.12,0.96)},anchor=north, draw=none,fill=none}
	, legend columns=1
	, width=1\textwidth
	, height=0.5\textwidth
	, ytick={0,150,300,450,600}
	, major tick length=2pt
	]
	\addplot [mark size=3pt, color=blue, mark=star] [unbounded coords=jump] table[col sep=semicolon, y index=2] {./cactusaspcomp.csv};
	\addlegendentry{\waspict}
	\addplot [mark size=3pt, color=green!50!black, mark=triangle] [unbounded coords=jump] table[col sep=semicolon, y index=3] {./cactusaspcomp.csv}; 
	\addlegendentry{\waspor}
	
	\addplot [mark size=3pt, color=orange, mark=o] [unbounded coords=jump] table[col sep=semicolon, y index=5] {./cactusaspcomp.csv};
	\addlegendentry{\waspopt}
	
	\addplot [mark size=3pt, color=black, mark=x] [unbounded coords=jump] table[col sep=semicolon, y index=4] {./cactusaspcomp.csv};
	\addlegendentry{\waspone}
	
	\addplot [mark size=3pt, color=yellow!50!black, mark=square] [unbounded coords=jump] table[col sep=semicolon, y index=6] {./cactusaspcomp.csv};
	\addlegendentry{\waspcm}
	
	\addplot [mark size=3pt, color=red, mark=pentagon] [unbounded coords=jump] table[col sep=semicolon, y index=1] {./cactusaspcomp.csv}; 
	\addlegendentry{\clasp}
	\end{axis}
	\end{tikzpicture}	
	\caption{Performance comparison on computation of cautious consequences for \emph{easy} instances of ASP Competitions: number of solved instances within a given per-instance time limit.}\label{fig:cactus:2}
	\figrule
\end{figure}
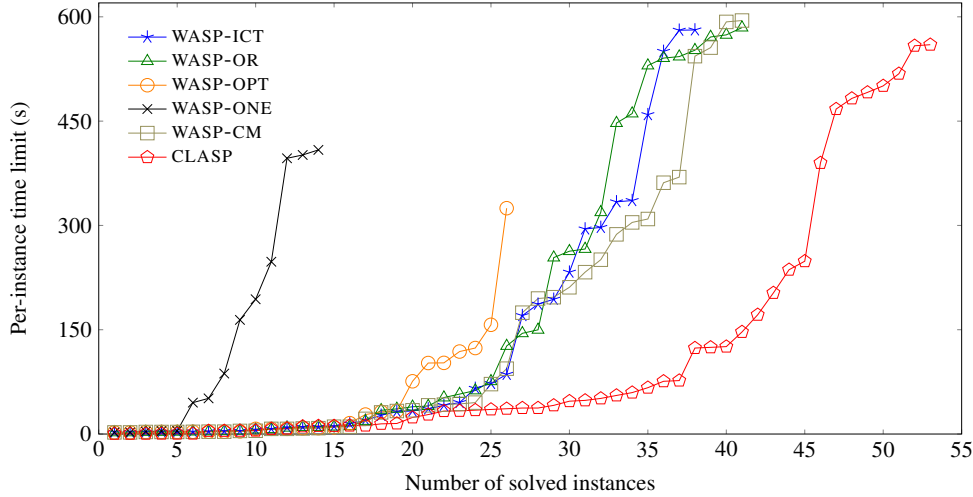

\subsection{Additional experiments}
Additional experiments are obtained from \emph{easy} instances of ASP Competitions, and from abstract argumentation frameworks by adapting an encoding used by \textsc{pyglaf} \cite{DBLP:conf/iclp/Alviano17}.
The performance gain of the new algorithms is less marked on these instances, either because the previous algorithms are already very efficient (abstract argumentation frameworks), or because only few atoms are actually cautious consequences of combinatorial problems in ASP Competitions.
\begin{table}[b!]
	\caption{
		Solved instances and cumulative running time (in seconds; each timeout adds 600 seconds) on computation of cautious consequences for \emph{easy} instances of ASP Competitions.
		Best performance emphasized in bold (within this respect, differences up to 10 seconds are ignored).
	}\label{tab:competition}
	\centering
	\tabcolsep=0.040cm
	\begin{tabular}{lrrrrrrrrrrrrr}
		\toprule
		& & \multicolumn{2}{c}{\textbf{\waspor}}	&\multicolumn{2}{c}{\textbf{\waspict}} &\multicolumn{2}{c}{\textbf{\waspopt}} &\multicolumn{2}{c}{\textbf{\waspone}} &\multicolumn{2}{c}{\textbf{\waspcm}} &\multicolumn{2}{c}{\textbf{\clasp}}\\
		\cmidrule{3-4}\cmidrule{5-6}\cmidrule{7-8}\cmidrule{9-10}\cmidrule{11-12}\cmidrule{13-14}
		\textbf{Benchmark} &  \textbf{\#} &	\textbf{sol}.	&	\textbf{sum t}	&	\textbf{sol.}	&	\textbf{sum t}	&	\textbf{sol.}	&	\textbf{sum t}	&	\textbf{sol.}	&	\textbf{sum t}	&	\textbf{sol.}	&	\textbf{sum t}	&	\textbf{sol.}	&	\textbf{sum t}\\
		\cmidrule{1-14}
		GracefulGraphs	&	1	&	1	&	149.9	&	1	&	85.8	&	1	&	\textbf{32.0}	&	1	&	45.2	&	1	&	45.2	&	1	&	51.2\\
		GraphCol	&	1	&	0	&	600.0	&	0	&	600.0	&	0	&	600.0	&	0	&	600.0	&	0	&	600.0	&	0	&	600.0\\
		IncrSched	&	6	&	4	&	1904.0	&	1	&	3003.4	&	1	&	3015.7	&	1	&	3164.2	&	2	&	2691.6	&	\textbf{5}	&	\textbf{857.0}\\
		KnightTour	&	2	&	0	&	1200.0	&	0	&	1200.0	&	0	&	1200.0	&	0	&	1200.0	&	0	&	1200.0	&	\textbf{2}	&	\textbf{626.1}\\
		Labyrinth	&	32	&	1	&	19184.5	&	0	&	19200.0	&	0	&	19200.0	&	0	&	19200.0	&	0	&	19200.0	&	\textbf{6}	&	\textbf{18377.2}\\
		NoMystery	&	2	&	\textbf{1}	&	\textbf{657.4}	&	1	&	672.7	&	0	&	1200.0	&	1	&	1001.3	&	1	&	694.0	&	1	&	1091.3\\
		PPM	&	15	&	15	&	88.7	&	15	&	94.4	&	15	&	\textbf{75.6}	&	10	&	4350.7	&	15	&	\textbf{80.8}	&	15	&	264.0\\
		QualSpatReas	&	18	&	13	&	6569.2	&	15	&	4893.3	&	7	&	7406.4	&	0	&	10800.0	&	17	&	4537.4	&	\textbf{18}	&	\textbf{1018.6}\\
		Sokoban	&	36	&	\textbf{4}	&	\textbf{20353.6}	&	3	&	20859.8	&	1	&	21102.1	&	1	&	21051.3	&	3	&	20665.2	&	3	&	20529.0\\
		VisitAll	&	2	&	2	&	529.0	&	2	&	364.2	&	1	&	757.1	&	0	&	1200.0	&	2	&	407.8	&	\textbf{2}	&	\textbf{80.2}\\
		\cmidrule{1-14}
		\textbf{Total}	&	115	&	41	&	51236.2	&	38	&	50973.5	&	26	&	54588.9	&	14	&	62612.7	&	41	&	50122.0	&	\textbf{53}	&	\textbf{43494.7}\\
		\bottomrule
	\end{tabular}
\end{table}

Concerning the instances from 7th ASP Competition, the benchmark set comprises instances classified as easy, that is, those for which a stable model is found within 20 seconds of computation by mainstream ASP systems;
these instances were selected for the computationally more demanding task of enumerating all cautious consequences of the programs.
However, here one should note that these programs are no \emph{per se} meant to be queried, and therefore the benchmark set is to be considered a crafted one, in contrast to the benchmarks already reported on, which have a clear semantical interpretation in terms of specific applications.

Results are given in the cactus plot of Figure~\ref{fig:cactus:2} and in Table~\ref{tab:competition}.
On the crafted benchmark set, \waspor and \waspcm performed better than \waspict, \waspone, and \waspopt.
Indeed, the best performance overall was observed for \waspcm, which solved 41 out of 115 instances with an average running time of 139.6 seconds, while \waspor solved the same number of instances but with an higher average running time of 166.7 seconds.
A slightly worse performance was reached by \waspict, which solved only 3 instances less than \waspor and \waspcm, while the algorithms based on \textsc{min}, i.e., \waspone and \waspopt, were less effective, and solved respectively 27 and 15 less instances than \waspcm and \waspor.
Moreover, the gap between \waspor and \clasp is observed also on these instances, with \clasp solving 12 instances more than \waspor.
All in all, the benchmark provides a positive result for \textsc{cm}.

The other benchmark is obtained by considering instances of skeptical acceptance under complete extensions in the context of abstract argumentation frameworks.
The 2nd International Competition on Computational Models of Argumentation (ICCMA'17) provides a broad set of 350 test cases originating from different contexts;
instances can be downloaded from \url{http://argumentationcompetition.org/2017}.
Complete extensions can be computed by adapting the propositional encoding implemented by \textsc{pyglaf} to the language of ASP as follows:
\begin{verbatim}
{in(X)} :- arg(X).                            % guess an extension
:- att(X,Y), in(X), in(Y).                    % conflict-free
attacked(X) :- att(Y,X), in(Y).               % attacked relation
:- att(Y,X), in(X), not attacked(Y).          % admissible
:- arg(X); attacked(Y) : att(Y,X); not in(X). % complete
#show.
#show in/1.
\end{verbatim}
Results are given in the cactus plot of Figure~\ref{fig:cactus:3} and in Table~\ref{tab:argumentation}.
All algorithms behave similarly on these instances, which are easily solved by \textsc{wasp} and \textsc{clasp}.

\begin{figure}[t]
	\figrule
	\begin{tikzpicture}[scale=0.85]
	\pgfkeys{%
		/pgf/number format/set thousands separator = {}}
	\begin{axis}[
	scale only axis
	, font=\normalsize
	, x label style = {at={(axis description cs:0.5,0.0)}}
	, y label style = {at={(axis description cs:0.0,0.5)}}
	, xlabel={Number of solved instances}
	, ylabel={Per-instance time limit (s)}
	, xmin=300, xmax=355
	, ymin=0, ymax=620
	, legend style={at={(0.12,0.96)},anchor=north, draw=none,fill=none}
	, legend columns=1
	, width=1\textwidth
	, height=0.5\textwidth
	, ytick={0,150,300,450,600}
	, major tick length=2pt
	]
	\addplot [mark size=3pt, color=blue, mark=star] [unbounded coords=jump] table[col sep=semicolon, y index=2] {./cactusargumentation.csv};
	\addlegendentry{\waspict}
	\addplot [mark size=3pt, color=green!50!black, mark=triangle] [unbounded coords=jump] table[col sep=semicolon, y index=3] {./cactusargumentation.csv}; 
	\addlegendentry{\waspor}
	
	\addplot [mark size=3pt, color=orange, mark=o] [unbounded coords=jump] table[col sep=semicolon, y index=5] {./cactusargumentation.csv};
	\addlegendentry{\waspopt}
	
	\addplot [mark size=3pt, color=black, mark=x] [unbounded coords=jump] table[col sep=semicolon, y index=4] {./cactusargumentation.csv};
	\addlegendentry{\waspone}
	
	\addplot [mark size=3pt, color=yellow!50!black, mark=square] [unbounded coords=jump] table[col sep=semicolon, y index=6] {./cactusargumentation.csv};
	\addlegendentry{\waspcm}
	
	\addplot [mark size=3pt, color=red, mark=pentagon] [unbounded coords=jump] table[col sep=semicolon, y index=1] {./cactusargumentation.csv}; 
	\addlegendentry{\clasp}			
	\end{axis}
	\end{tikzpicture}	
	\caption{Performance comparison on skeptical acceptance under complete extensions: number of solved instances within a given per-instance time limit.}\label{fig:cactus:3}
	\figrule
\end{figure}
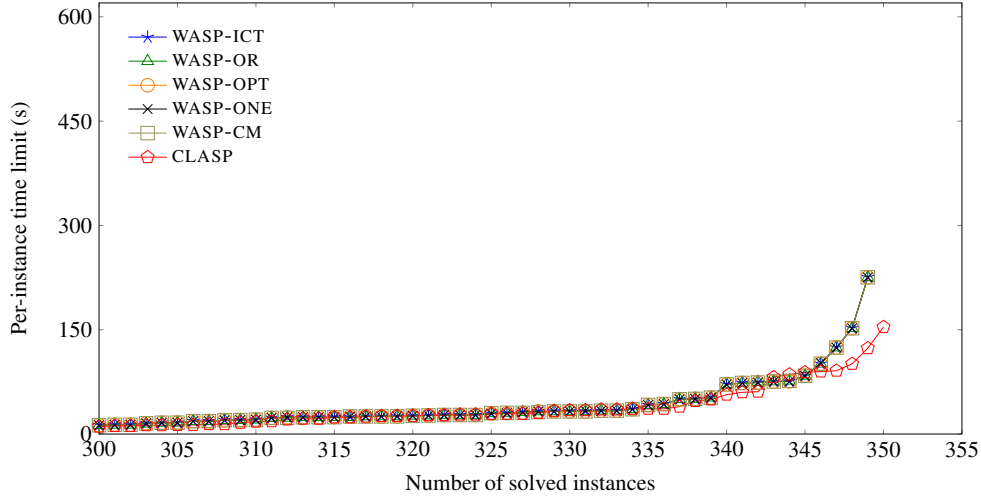

\begin{table}[!b]
	\caption{
		Solved instances and cumulative running time (in seconds; each timeout adds 600 seconds) on skeptical acceptance under complete extensions.
		Best performance emphasized in bold (within this respect, differences up to 10 seconds are ignored).
	}\label{tab:argumentation}
	\centering
	\tabcolsep=0.058cm
	\begin{tabular}{lrrrrrrrrrrrrr}
		\toprule
		& & \multicolumn{2}{c}{\textbf{\waspor}}	&\multicolumn{2}{c}{\textbf{\waspict}} &\multicolumn{2}{c}{\textbf{\waspopt}} &\multicolumn{2}{c}{\textbf{\waspone}} &\multicolumn{2}{c}{\textbf{\waspcm}} &\multicolumn{2}{c}{\textbf{\clasp}}\\
		\cmidrule{3-4}\cmidrule{5-6}\cmidrule{7-8}\cmidrule{9-10}\cmidrule{11-12}\cmidrule{13-14}
		\textbf{Benchmark} &  \textbf{\#} &	\textbf{sol}.	&	\textbf{sum t}	&	\textbf{sol.}	&	\textbf{sum t}	&	\textbf{sol.}	&	\textbf{sum t}	&	\textbf{sol.}	&	\textbf{sum t}	&	\textbf{sol.}	&	\textbf{sum t}	&	\textbf{sol.}	&	\textbf{sum t}\\
		\cmidrule{1-14}
		C -- Dataset 1	& 50 &	50	&	62.5	&	50	&	62.4	&	50	&	62.6	&	50	&	62.4	&	50	&	62.4	&	50	&	\textbf{18.7}\\
		C -- Dataset 2	& 50 &	50	&	553.1	&	50	&	552.8	&	50	&	553.3	&	50	&	553.9	&	50	&	553.4	&	50	&	\textbf{310.3}\\
		C -- Dataset 3	& 100 &	100	&	\textbf{628.6}	&	100	&	\textbf{629.1}	&	100	&	\textbf{628.5}	&	100	&	\textbf{630.4}	&	100	&	\textbf{627.6}	&	100	&	650.7\\
		C -- Dataset 4	& 150 &	149	&	1788.1	&	149	&	1785.8	&	149	&	1785.7	&	149	&	1786.6	&	149	&	1787.1	&	\textbf{150}	&	\textbf{1184.1}\\
		\cmidrule{1-14}
		\textbf{Total}	& 350 &	349	&	3032.3	&	349	&	3030.1	&	349	&	3030.1	&	349	&	3033.3	&	349	&	3030.5	&	\textbf{350}	&	\textbf{2163.8}\\
		\bottomrule
	\end{tabular}
\end{table}

\section{Related work}
\label{sec:related}

Cautious reasoning over answer set programs is the underlying computational task of several practical applications, among them Multi Context Systems Querying (MCSQ; \citeNP{DBLP:conf/aaai/BrewkaE07}) and Consistent Query Answering (CQA; \citeNP{DBLP:journals/tplp/ArenasBC03}).
MCSQ allows to query heterogeneous knowledge bases, called contexts, linked together by bridge rules modeling the flow of information among contexts.
CQA, instead, allows to query inconsistent databases by taking advantage of the notion of repair, that is, a maximal and consistent revision of the database \cite{DBLP:journals/tplp/ArenasBC03,DBLP:journals/tplp/MannaRT13,DBLP:journals/tplp/MannaRT15}.
Both benchmarks were used in Section~\ref{sec:exp}.

In ASP, cautious reasoning has been previously  implemented in
\textsc{dlv} \cite{DBLP:journals/tocl/LeonePFEGPS06,DBLP:conf/datalog/AlvianoFLPPT10,DBLP:conf/lpnmr/AlvianoCDFLPRVZ17}, 
\clasp \cite{DBLP:journals/ai/GebserKS12}, and 
\wasp \cite{DBLP:conf/lpnmr/AlvianoDLR15}.
These systems implement dedicated techniques built on 
top of their search procedures for computing stable models. The algorithm implemented in 
\textsc{dlv} is based on the enumeration of all stable models for computing their intersection.
\clasp instead implements \textsc{or}, and \wasp implements both \textsc{or} and \textsc{ict} \cite{DBLP:journals/tplp/AlvianoDR14}.
An abstract
framework of cautious reasoning, capturing the above-mentioned algorithms,
has been presented 
by \citeANP{DBLP:conf/iclp/BrocheninM15}~\citeyear{DBLP:conf/iclp/BrocheninM15}. 

The task of computing the set of cautious consequences of a program $\Pi$ is closely
related to computing the backbone of a propositional theory
\cite{DBLP:journals/aicom/JanotaLM15,DBLP:conf/hldvt/ZhuWSM11}. Some of the
algorithms presented  are inspired by previous work in the
context of propositional logic. Algorithm \textsc{or} is based on a strategy for computing
the backbone of a conjunctive normal formal propositional formula~\cite{DBLP:conf/hldvt/ZhuWSM11}, and extended with some
optimization techniques by \citeANP{DBLP:journals/aicom/JanotaLM15}~\citeyear{DBLP:journals/aicom/JanotaLM15}. Algorithm \textsc{ict} is a
reimplementation of Algorithm~3 from \citeANP{DBLP:journals/aicom/JanotaLM15}~\citeyear{DBLP:journals/aicom/JanotaLM15}.

Algorithm \textsc{min} captures a family of strategies based on the computation of stable 
models that are subset-minimal with respect to a set of atoms. In particular, the approach to computing 
subset-minimal stable models adapts the 
\textsc{optsat} algorithm \cite{DBLP:journals/constraints/RosaGM10,DBLP:journals/amai/GiunchigliaLM08} which has
been previously employed 
for the computation of optimal stable models \cite{Alviano01092015} in the context of ASP, 
as well as for computing the ideal semantics in the context of abstract 
argumentation \cite{PrevitiSAC18}; and another approach is based
on algorithm \textsc{one}, which was introduced for the first time as a MaxSAT 
algorithm \cite{DBLP:conf/ijcai/AlvianoDR15}. Stable models that are 
subset-minimal with respect to a set of atoms have also been previously employed 
in the computation of 
paracoherent answer sets \cite{DBLP:conf/aaai/AmendolaDFLR17}, and for 
circumscription \cite{DBLP:journals/tplp/Alviano17}. Other strategies for the 
computation of minimal models have been previously proposed in early work
\cite{DBLP:conf/tableaux/Niemela96,DBLP:conf/cade/HasegawaFK00,DBLP:journals/jar/BryY00}
preceding the invention of conflict-driven solvers;
this limitation was overcome by \citeANP{Koshimura09}~\citeyear{Koshimura09}. 
Going beyond the above, 
general approaches such as the 
algorithms for computing minimal sets over monotone 
predicates (MSMP; \citeNP{DBLP:journals/ai/JanotaM16}) could also be adapted for 
computing minimal models. 

Although \textsc{cm} is a novel algorithm, it shares some ideas with 
Algorithm~6 by \citeANP{DBLP:journals/aicom/JanotaLM15}~\citeyear{DBLP:journals/aicom/JanotaLM15}. In particular, both algorithms try 
to flip all candidates at once. In case a core of size one is returned, both algorithms add 
the single atom (literal) to the set of cautious consequences (backbone). However, 
the two algorithms differ in their behavior when a core with size greater than one is returned. The 
algorithm presented by \citeANP{DBLP:journals/aicom/JanotaLM15} simply stores the 
literals in the core in a new set, and before termination processes them with an 
alternative algorithm.
Instead, \textsc{cm} starts a minimization procedure of the returned core.
If the size of the reduced core is one, then a cautious consequence has been found.
Otherwise, during the minimization, the computed stable model is used to improve the overestimation. 
Moreover, algorithm \textsc{cm} is also \textit{anytime} by itself, since it is able to improve the 
underestimation during its computation, while algorithm \textsc{min} can be made anytime using 
the technique suggested by \citeANP{DBLP:journals/tplp/AlvianoDR14}~\citeyear{DBLP:journals/tplp/AlvianoDR14}.

As a final remark, all considered algorithms work on ground programs.
It turns out that they can be used as a back end of query optimization techniques that work at the symbolic level, as for example magic sets \cite{DBLP:journals/aicom/AlvianoF11,DBLP:journals/ai/AlvianoFGL12}.

\section{Conclusion}

Cautious reasoning over answer set programs is a central query answering task supported by ASP systems and is motivated by various applications.
The main contributions of this paper are two new algorithms for cautious reasoning in ASP, their implementation, and empirical evaluation against earlier proposed approaches to cautious reasoning in ASP.
Actually, one of the new algorithms make use of minimal models, and can be instantiated with several procedures, two of them considered in this paper.
The other algorithm harnesses the ability of ASP solvers to extract unsatisfiable 
cores for expediting the elimination of candidate atoms that are not cautious consequences. 
The empirical results show that the proposed algorithms improve the current state of the art 
in cautious reasoning on standard benchmarks with non-ground queries from the latest ASP competitions.

\section{Acknowledgments}

Mario Alviano was partially supported by the POR CALABRIA FESR 2014-2020 project ``DLV Large Scale'' (CUP J28C17000220006), by the EU H2020 PON I\&C 2014-2020 project ``S2BDW'' (CUP B28I17000250008), and by Gruppo Nazionale per il Calcolo Scientifico (GNCS-INdAM). 
Carmine Dodaro and Marco Maratea were  partially  supported  by  project  GESTEC - Service
Oriented  Technologies  for  Integrated  ICT  platforms 
(Italian Ministry for Education Research and University, DM 64565).
Matti J\"arvisalo and Alessandro Previti were supported  by Academy of Finland (grants 276412 and 312662).

\bibliographystyle{acmtrans}

\label{lastpage}

\end{document}